\documentclass[prl,reprint,amsmath,amssymb,aps,floatfix,longbibliography,superscriptaddress]{revtex4-2}
\usepackage[utf8]{inputenc}
\usepackage{graphicx}
\usepackage{amsfonts}
\usepackage{hyperref}
\usepackage{amsmath}
\usepackage{amsthm}

\newtheorem{theo}{Tplottin ubuntuheorem}
\theoremstyle{plain}
\newtheorem{thm}[theo]{Theorem}
\newtheorem{lem}[theo]{Lemma}
\newtheorem{prop}[theo]{Proposition}

\newtheorem*{thm*}{Theorem}
\newtheorem*{lem*}{Lemma}
\newtheorem*{prop*}{Proposition}
\newtheorem*{cor*}{Corollary}

\theoremstyle{definition}
\newtheorem{defn}[theo]{Definition}

\newtheorem{rem}[theo]{Remark}

\begin{document}
\title{Topologically Protected Vortex Knots in an Experimentally Realizable System}

\author{Hermanni Rajamäki}
\affiliation{QCD Labs, QTF Centre of Excellence, Department of Applied Physics, Aalto University, P.O. Box 13500, FI-00076 Aalto, Finland}

\author{Toni Annala}
\email{tannala@ias.edu}
\affiliation{QCD Labs, QTF Centre of Excellence, Department of Applied Physics, Aalto University, P.O. Box 13500, FI-00076 Aalto, Finland}
\affiliation{School of Mathematics, Institute for Advanced Study, 1 Einstein Drive, Princeton, 08540, NJ, USA}

\author{Mikko Möttönen}
\affiliation{QCD Labs, QTF Centre of Excellence, Department of Applied Physics, Aalto University, P.O. Box 13500, FI-00076 Aalto, Finland}


\date{\today}

\begin{abstract}
Ordered media often support vortex structures with intriguing topological properties. Here, we investigate non-Abelian vortices in tetrahedral order, which appear in the cyclic phase of spin-2 Bose--Einstein condensates and in the tetrahedratic phase of bent-core nematic liquid crystals. Using these vortices, we construct topologically protected knots in the sense that they cannot decay into unlinked simple loop defects through vortex crossings and reconnections without destroying the phase. The discovered structures are the first examples of knots bearing such topological protection in a known experimentally realizable system.
\end{abstract}

\maketitle

\setcounter{section}{0}

Topological phenomena in physics have attracted persistent attention since Lord Kelvin's vortex-atom theory in 1869~\cite{thomson:1869}, which hypothesized atoms to be knotted vortices in the ether that permeates throughout the cosmos. Since then, interesting topological structures have been discovered in various other contexts: the topologically non-trivial gauge field theory of the Dirac monopole~\cite{dirac:1931}, topological defects~\cite{volovik:1977, mermin:1979} that appear in various ordered systems such as superfluids~\cite{volovik:2009}, Bose--Einstein condensates (BEC)~\cite{kawaguchi:2012,ueda:2014}, liquid crystals~\cite{machon:2013,beller:2014,machon:2016b,machon:2019}, and the vacuum structures of quantum field theories~\cite{bucher:1992, bucher:1992b, bucher:1994}, as well as topological phases of matter such as topological superconductors~\cite{sato:2017} and topological insulators~\cite{hasan:2010}. Interestingly, such phenomena are amenable to investigation by numerical~\cite{faddeev:1997,ruostekoski:2003}, and experimental means~\cite{ray:2014,ray:2015,hall:2016, lee:2018}, resulting in a discipline in which practical experiments and numerical simulations meet methods of abstract topology.

Lord Kelvin's hypothesis on atoms was inspired by the mathematical observation that, in a frictionless ideal fluid, vortex cores cannot pass through each other~\cite{thomson:1869}, resulting in a conservation of the knotting type, and a related conserved quantity \emph{helicity}~\cite{moffatt:1969}. Thus, knots tied from vortices in such a system behave equivalently to physical knots. In contrast, in ordinary fluids and superfluids, the knot structure of a vortex configuration tends to decay because of allowed strand crossings and local reconnections~\cite{kleckner:2013,kleckner:2014,kleckner:2016}. 

However, \textit{topological vortices} in ordered media provide a setting with intermediate stability properties since there exist topological obstructions for strand crossings~\cite{poenaru:1977} and local reconnections~\cite{annala:2022}. Moreover, even though the local structure of topological defects is well-understood in terms of the \emph{homotopy groups} of the \emph{order parameter space}~\cite{mermin:1979}, as is the global structure of ring defects and unlinked configurations thereof~\cite{nakanishi:1988, annala-mottonen}, the global structure of topological vortex knots and links is more subtle~\cite{machon:2014, machon:2016, annala:2022}. A heuristic explanation for this challenge is provided by the phenomenon of \emph{topological influence}, namely, that the topological charge of a topological defect may be altered as it traverses about the core of a topological vortex~\cite{kobayashi:2012, kobayashi:2014, annala-mottonen}. A topological vortex knot traverses about itself in a non-trivial fashion, resulting in a potentially intricate topological structure.

The topological protection of those topological vortex configurations in the biaxial-nematic and cyclic phases of spin-2 BECs that have no scalar-phase winding about the core, as well as all vortex configurations in the biaxial-nematic liquid crystal, were investigated in Ref.~\cite{annala:2022}. Unfortunately, such vortices do not support any topologically protected knots since a link of at least three components is necessary to achieve topological protection and such configurations reduce into simpler structures by topologically allowed strand crossings and reconnections. This requirement of three components is related to the structure of the \emph{quaternion group} 
\begin{equation}
    Q_8 := \{\pm 1, \pm i, \pm j, \pm k \},
\end{equation}
which controls the behavior of the particular class of vortex structures, there is no general mathematical obstruction for the existence of topological vortex knots. In fact, topological vortices controlled by the \emph{symmetric group} $\Sigma_3$ on three elements exhibit at least two types of topological vortex knots: the left- and the right-handed tricolored trefoil knot~\cite{annala2023}. Unfortunately, no experimentally achievable physical systems which support such vortices has been found.

In this paper, we discover the first topologically protected knot structure in an experimentally realizable order parameter field, namely, in the cyclic phase of a spin-2 BEC. The behavior of such topological vortices is essentially controlled by the \emph{binary tetrahedral group}
\begin{equation}
    T^* := \left\{\pm 1, \pm i, \pm j, \pm k, \tfrac{\pm 1 \pm i \pm j \pm k}{2} \right\},
\end{equation}
which contains 24 elements and the quaternion group $Q_8$ as a subgroup. In physical terms, the subgroup $Q_8$ corresponds to the subclass of topological vortices in the cyclic phase of the spin-2 BEC, whereas the group $T^*$ describes the behavior of also those vortices with fractional scalar-phase winding~\cite{semenoff:2007, turner:2007}. Interestingly, the binary tetrahedral group also describes topological vortices in the tetrahedratic phase of bent-core nematic liquid crystals~\cite{lubensky:2002}.

Next, we recall the formalism of topological vortices. Let $M$ denote the order parameter space of an ordered medium and $X \subset \mathbb{R}^3$ describe the physical extent of the system. A \emph{topological vortex configuration} consists of a \emph{core} $\mathcal{L} \subset X$, and a continuous map 
\begin{equation}
    \Psi: X \backslash \mathcal{L} \to M
\end{equation}
referred to as the \emph{order parameter field}. The core of the defect $\mathcal{L} \subset X$ is assumed to consist of smoothly embedded circles. In other words, $\mathcal{L}$ is a \emph{tame link}~\cite{rolfsen:2003}. More general one-dimensional core structures, such as embedded graphs, are not considered here for simplicity. 

As in Ref.~\cite{annala:2022}, we consider the evolution of the vortex configuration with three types of processes: (i) topology-preserving evolution where the defect-core evolves up to smooth isotopy, and the order parameter field evolves up to homotopy, (ii) \emph{topologically allowed strand crossing}~\cite{poenaru:1977} where two cores may cross if they correspond to commuting elements in the \emph{fundamental group} $\pi_1(M)$ of the order parameter space, and (iii) \emph{topologically allowed reconnection}~\cite{annala:2022} where two vortex cores may reconnect locally if there is no topological obstruction for it. Collectively, the two topology-altering processes are referred to as \emph{topologically allowed modifications}. A topological vortex structure that cannot evolve into a collection of unlinked simple loops via the processes explained above is referred to as being \emph{topologically protected}.

However, the following three types of processes are not considered in our definition of topological protection: (i) non-commuting elements crossing, and forming a rung vortex between the crossing points~\cite{poenaru:1977}, (ii) loops vanishing by shrinking into a point, and (iii) simple vortices splitting along a significant distance. By simple vortices, we mean non-multiply charged vortices in the sense that the topological charge cannot be expressed as a product of simpler elements of the fundamental group. For example, if the fundamental group is $\mathbb{Z}$, then vortices corresponding to $\pm 1$ are simple, and for the groups $Q_8$ and $T^*$, we regard elements that are different from $\pm 1$ as being simple. We physically justify the omission of such events by focusing on low-energy states in low-temperature systems, where such processes are energetically unfavorable. To this end, we also note the following: (i) the energy required to for a rung vortex is linearly proportional to its length in general rendering the formation of a vortex cord long enough to affect the evolution of the structure unlikely, (ii) in certain BEC systems, ring vortices are protected by an energy barrier against shrinking into a point defect~\cite{ruostekoski:2003}, and (iii) the splitting of vortices is energetically unfavourable in some biaxial-nematic liquid crystal systems, and is not observed in simulations~\cite{priezjev:2002}.

Henceforward, we assume that our physical system is a finite-sized spherical sample of an ordered medium, and that none of the topological vortices pierce the boundary of the sample. In mathematical terms, the physical extent of the system $X \subset \mathbb{R}^3$ is modelled by the closed unit ball
\begin{equation}
\bar{\mathbb{B}}^3 := \{\mathbf{x} \in \mathbb{R}^3 : \vert\vert \mathbf{x} \vert\vert \leq 1\} \subset \mathbb{R}^3,   
\end{equation}
and the core link $\mathcal{L}$ is contained in the interior of $\bar{\mathbb{B}}^3$. 
Under the assumption that the \textit{second homotopy group} $\pi_2(M)$ of the order parameter space $M$ is trivial, a topological vortex configuration may be faithfully represented by a \textit{$\pi_1(M)$-colored link} (Fig.~\ref{fig:1}), which is a combinatorial object that captures both the topology of the core configuration $\mathcal{L}$, and the homotopy class of the order parameter field $\Psi: X \backslash \mathcal{L} \to M$~\cite{annala:2022}. 

The order parameter space for the cyclic phase of spin-2 BEC is 
\begin{equation}
    M_\mathrm{C} = [S^1 \times \mathrm{SO}(3)] / T,
\end{equation}
where the circle $S^1$ is regarded as the unit complex numbers, and $T$ is the group of the rotational symmetries of a tetrahedron~\cite{semenoff:2007, turner:2007}. Moreover, the fundamental group of $M_\mathrm{C}$ may be realized as a subgroup of $\mathbb{Z} \times T^*$ where the first component records the fractional scalar-phase winding in multiples of $1/3$. Moreover, for every element of $T^*$ there exists such a choice for the first parameter that the pair is contained in $\pi_1(M_\mathrm{C})$, and the choice is unique modulo 3. Thus, any $T^*$-colored link may be constructed as a topological vortex configuration in spin-2 cyclic BEC. Importantly, the theory of $T^*$-colored links may be employed to investigate the existence of topologically protected vortex structures in spin-2 cyclic BEC.

Next, we investigate the topological properties of $T^*$-colored links. The $T^*$-colored link diagrams may be graphically presented in terms of bicolored links, as explained in Fig.~\ref{fig:TColors}. There are essentially three classes of bicolorings: those corresponding to the elements $\{\pm i, \pm j, \pm k\}$ which we refer to as $Q_8$ \emph{strands}, those corresponding to elements of the form $(1 \pm i \pm j \pm k)/2$ which we refer to as \textit{black strands}, and those corresponding to elements of the form $(-1 \pm i \pm j \pm k)/2$ which we refer to as \emph{gray strands}. The crossing rules in $T^*$-colored link diagrams shown in Fig.~\ref{fig:TCrossing} can be understood in terms of rotations of a decorated tetrahedron (see the Supplemental Material). Topologically allowed modifications admit a graphical presentation as well, as explained in Fig.~\ref{fig:modificationrules}.

In order to verify the topological protection of a vortex configuration, we construct invariants of $T^*$-colored links that are conserved in topologically allowed modifications, and for which a configuration of unlinked simple loops obtains a trivial value. Hence, a non-trivial value of the invariant implies topological protection.

We introduce two invariants, $Q_b$ and $Q_g$, which are computed using the black strands, and the gray strands, respectively. These invariants are similar to the $Q$-invariant defined in Ref.~\cite{annala:2022}.

The invariant $Q_g \in \mathbb{Z}_6$ for gray loops is defined by first interchanging the corresponding black and gray bicolorings, and then computing the invariant $Q_b \in \mathbb{Z}_6$ for black loops, the definition of which we explain next. Each black and gray loop has an orientation such that the black or the gray color of the bicoloring is on the right side when starting from the basepoint and traversing along the loop according to the orientation. For each black loop $L$ of the diagram, we choose a basepoint $b$ on one of the arcs of $L$ in the diagram, and define $q(L,b) \in T^*$ as the product of elements one obtains when $L$ crosses under other strands and itself, in the order the crossings occur along $L$, from right to left. The element $q(L,b)$ is a power of the element $\chi(b)$ corresponding to the coloring at the basepoint,
\begin{equation}
    q(L,b) = \chi(b)^{Q'(L)},
\end{equation}
and the exponent $Q'(L) \in \mathbb{Z}_6$ is independent of the choice of the basepoint, and is well-defined modulo 6 (see the Supplemental Material). The black invariant is then defined as
\begin{equation}
    Q_b := \sum_{L \text{ is a black loop}} Q'(L) - \omega_b - 4l_{bg} \in \mathbb{Z}_6,
\end{equation}
where $\omega_b$ is the signed count of crossings where both strands are black, and $l_{bg}$ is the signed count of crossings where a black strand crosses under a gray strand (Fig.~\ref{fig:Q'def}). The sign of each crossing is decided by the right-hand rule (see the Supplemental Material). 

The invariants $Q_b$ and $Q_g$ are conserved in topologically allowed local modifications (see the Supplemental Material). Clearly, these invariants vanish for a configuration of unlinked simple loops. A non-zero value of either invariant therefore implies a topologically protected structure: the link cannot decay into unlinked simple loops by the processes described above. However, we suspect that the invariants are not complete in the sense that structures with identical invariants may not necessarily be topologically equivalent as exemplified in Fig.~\ref{fig:examples}.

The topological invariants defined above demonstrate that topologically protected knots and links exist in the cyclic phase of spin-2 BEC. Two $T^*$-colored knot structures, the trefoil and the figure-eight knot, illustrated in Fig.~\ref{fig:examples}(a), have invariant $Q_b=[3]$. Equivalent knots can be constructed using grey loops, with $Q_g = [3]$. Fig.~\ref{fig:examples}(b) demonstrates the existence of a topologically protected link consisting of one black and one gray loop. The $T^*$-colored trefoil and figure-eight knots differ by a black Solomon's link structure which is a $T^*$-colored link consisting of two black loops. In the Supplemental Material, we show that every $T^*$-colored link structure consisting of only black loops decomposes into unlinked configurations of $T^*$-colored trefoil and figure-eight knots. Hence, the knots displayed in Fig.~\ref{fig:examples}(a) can be regarded as building blocks of purely black $T^*$-colored link structures consisting of only black strands. By symmetry, an equivalent statement is true for purely gray $T^*$-colored link structures. No such decomposition theorems are known for structures containing different types of loops.

In conclusion, we studied topologically intriguing vortex structures in the cyclic phase of spin-2 BEC, by employing the mathematical theory of $T^*$-colored links. This formalism also applies to bent-core nematic liquid crystals. We produced the first examples of topologically protected vortex knots supported by an ordered medium realizable in contemporary experiments.

Our results leave several interesting questions open, such as the detailed method to realize the proposed vortex structures in practice using experimentally available control techniques. This is a challenging task since topologically protected knots cannot, by definition, be manufactured by performing local surgeries to an unknotted loop defect, which is a method that has been applied successfully in the case of colloidal particles~\cite{tkalec:2011}. A rapid phase transition may generate a random vortex structure~\cite{sec:2014}, but often a vortex network rather than a vortex link is created~\cite{priezjev:2002}. Alternatively, controlled tailoring of electromagnetic or other fields that would induce the desired structure in the condensate, in the spirit of of~\cite{ray:2015, ollikainen:2017, pietila:2009b}, could be pursued. On the theoretical side, the complete classification $T^*$-colored links up to topological modifications remains open.

We thank T. Machon for pointing out that the tetrahedral order arises in bent-core nematic liquid crystals. We have received funding from the Academy of Finland Centre of Excellence program (Project No. 336810). T.A. thanks IAS for excellent working conditions.  

\bibliography{bibliography}

\newpage

\begin{figure}[h!]
    \centering
    \includegraphics[scale=1]{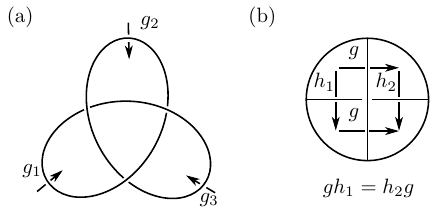}
    \caption{
    \textbf{$G$-colored link diagrams.}
    (a) Group homomorphism $\phi:\pi_1(\mathbb{R}^3$\textbackslash$L)\rightarrow G$ can be presented by a $G$-colored link diagram, in which each arc is decorated with an arrow and an element $g$ of $G$. The element records the value of $\psi[\gamma]$, where $[\gamma]\in \pi_1(\mathbb{R}^3\backslash L)$ is the homotopy class of a loop $\gamma$, which begins from a fixed basepoint above the plane of the picture, travels directly to the base of the arrow, traverses below the arc along the arrow, and returns back to basepoint. Loops which can be continuously deformed into each other, while keeping the basepoint fixed, belong to the same homotopy class. 
    (b) For the group homomorphism to be well-defined, the \textit{Wirtinger relation} must be satisfied at each crossing.}
    \label{fig:1}
\end{figure}

\begin{figure}[h!]
    \centering
    \includegraphics[scale=0.7]{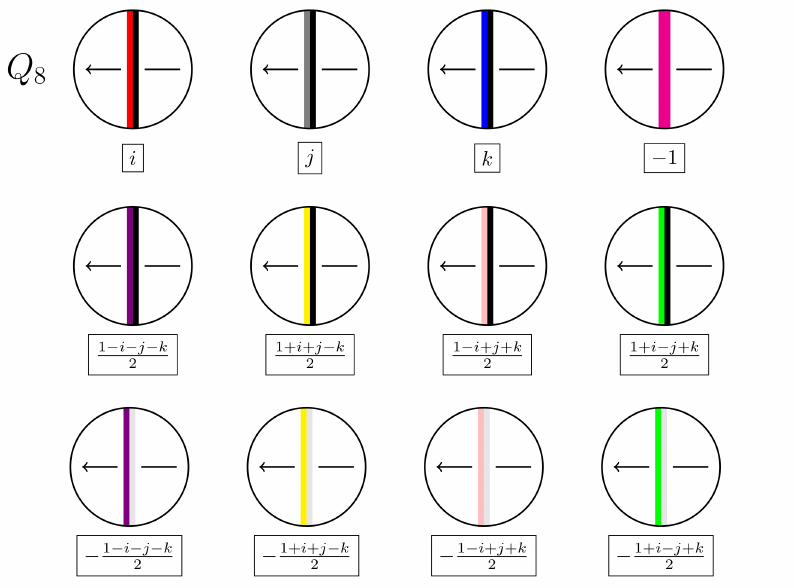}
    \caption{
    \textbf{$\mathbf{T^*}$-colored link diagrams: coloring scheme.}
    The cores of the topological vortices are represented by bicolored arcs, except the element $\{-1\}$, which is represented by a singly-colored arc. Crossing under the strand in the opposite direction corresponds to the inverse element in $T^*$, which coincides with the conjugate quaternion. The elements of the top and the middle row form conjugacy classes of four elements each. The elements of the bottom row, as well as their inverses, form the normal subgroup $Q_8 \subset T^*$. 
    }
    \label{fig:TColors}
\end{figure}

\begin{figure}[h!]
    \centering
    \includegraphics[scale=0.4]{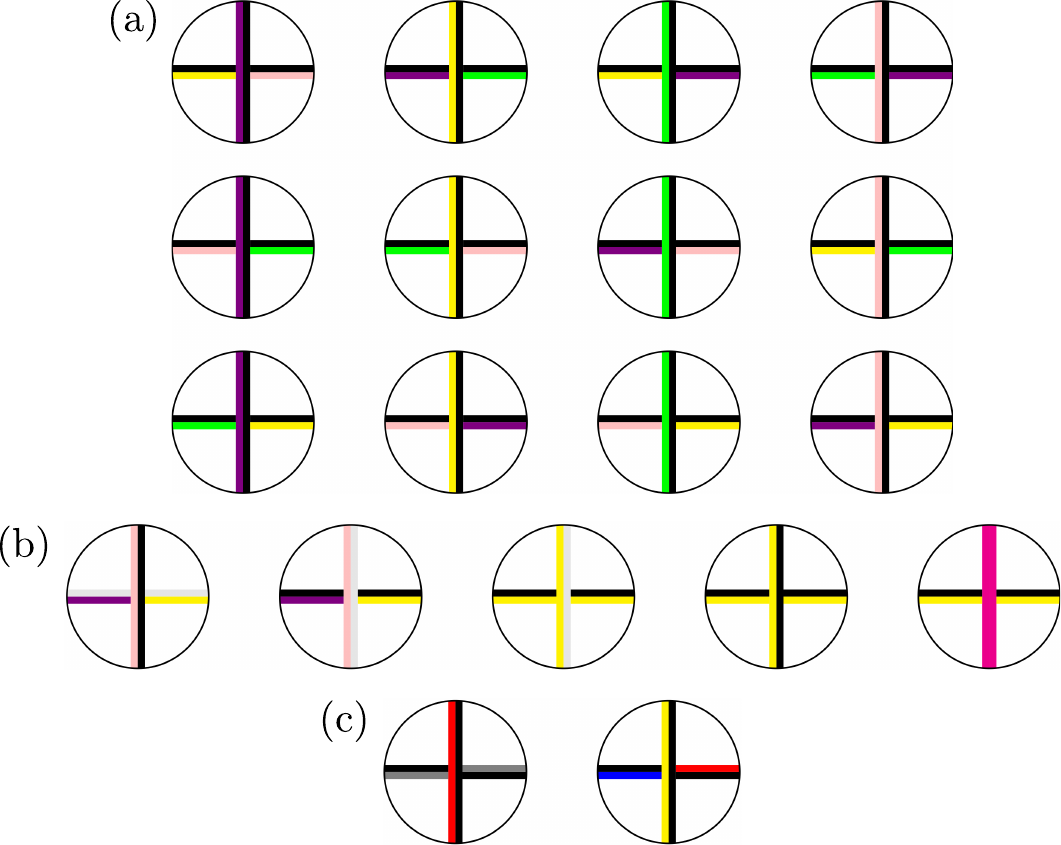}
    \caption{
    \textbf{$\mathbf{T^*}$-colored link diagrams: crossing rules according to the Wirtinger relation.}
    (a)~Complete crossing rules between the black strands. The color changes in the undercrossing, but the direction of the bicoloring remains unchanged. 
    (b) Passing under a strand with an identical colored part of the bicoloring, or a purple strand, does not change the bicoloring. A bicoloring with grey colored part behaves identically to the corresponding black bicoloring in crossings, because the group elements differ by the multiplication by the element $-1$ which commutes with every element of $T^*$. 
    (c) In a $Q_8$-colored diagram, the bicoloring is flipped when passing under another bicolored strand in $Q_8$. When passing under a strand that does not belong to the subgroup $Q_8$, the bicoloring does not necessarily flip and the color may change.
    }
    \label{fig:TCrossing}
\end{figure}

\begin{figure}[h!]
    \centering
    \includegraphics[scale=.85]{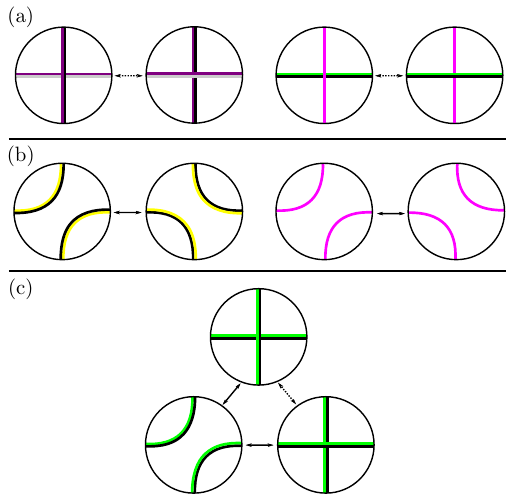}
    \caption{\textbf{Topologically allowed local modifications.} The two types of local modifications considered, strand crossings and reconnections, are marked by dashed and solid arrows respectively. 
    (a) Strand crossings may occur between strands having identical colored parts of the bicoloring, as well as between a purple strand and any other strand. 
    (b) Reconnections may only occur between strands with identical bicolorings. 
    (c) A crossing between the bicolored strands of the same kind may be modified with both the strand crossing and reconnection.}
    \label{fig:modificationrules}
\end{figure}

\begin{figure}[h!]
    \centering
    \includegraphics[scale=1.]{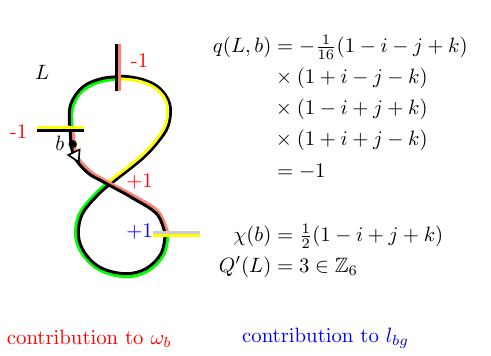}
    \caption{\textbf{Computation of $Q'(L)$ and $Q_b$.} 
    Example computation of the $Q'(L)$-invariant for a black loop $L$ in a $T^*$-colored link diagram. The contributions of the displayed crossings to $\omega_b$ and $l_{bg}$ are also marked in the figure by red and blue text, respectively.
    }\label{fig:Q'def}
\end{figure}

\begin{figure}[h!]
    \centering
    \includegraphics[scale=1]{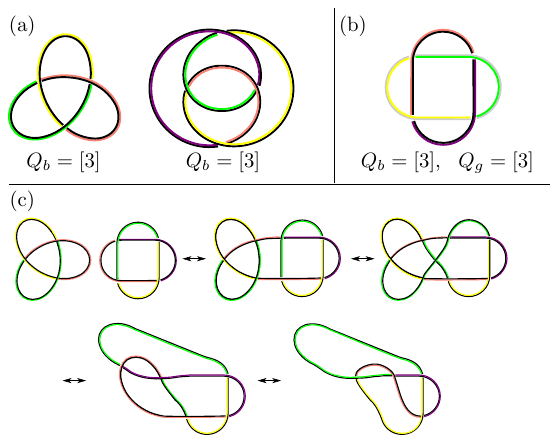}
    \caption{\textbf{Examples of topologically protected $\mathbf{T^*}$-colored knots and links.} 
    (a) Topologically protected $T^*$-colored trefoil knot and figure-eight knot and the values of their invariants. 
    (b) Topologically protected $T^*$-colored Solomon's link, tied from one black and one gray loop. 
    (c) A trefoil may be turned into a figure-eight knot by combining it with a $T^*$-colored Solomon's link consisting of two black loops. The black Solomon's link has trivial invariants, but it is unknown if it is topologically protected or not. Hence, even though the $T^*$-colored trefoil and the figure-eight knot have equivalent invariants, it is unknown if they can be turned into each other by topologically allowed modifications.
    }\label{fig:examples}
\end{figure}

\clearpage

\pagebreak

\widetext
\begin{center}
\textbf{\large Supplemental Material for: Topologically Protected Vortex Knots in an Experimentally Realizable System}
\end{center}

\section{Elements of the binary tetrahedral group as rotations of a decorated tetrahedron}

\begin{figure}[h!]
    \centering
    \includegraphics[scale=1]{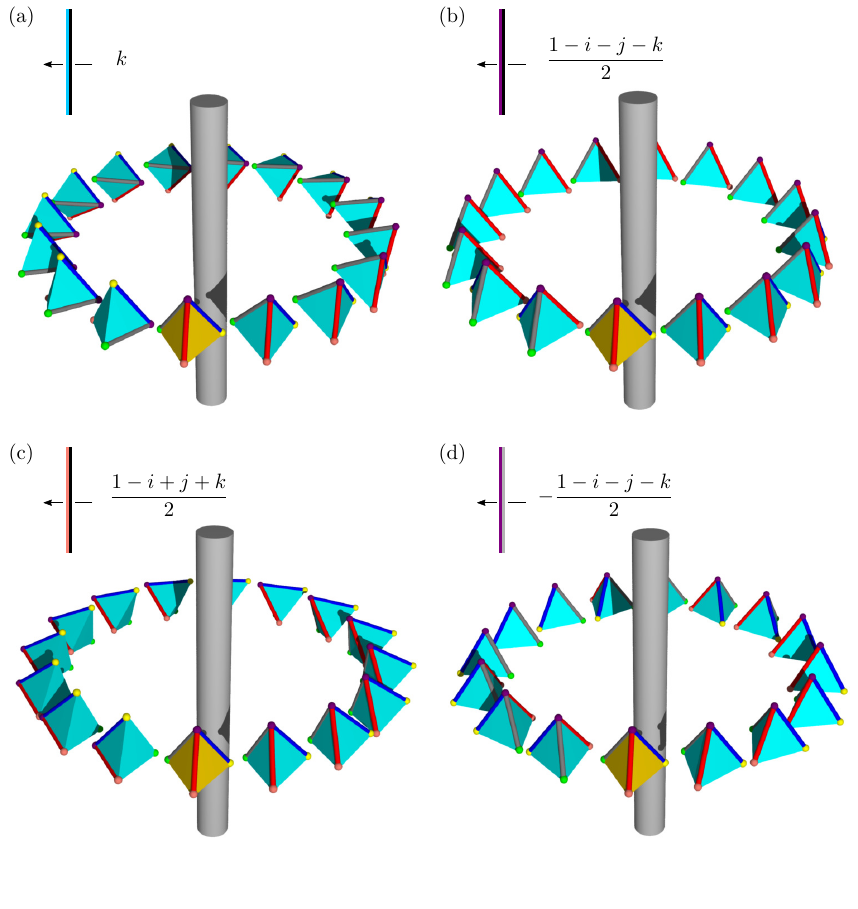}
    \caption{
    (a) Elements $i$, $j$, and $k$ correspond to rotating the tetrahedron by 180 degrees about the centre of the corresponding colored edge, according to the right hand rule. 
    (b), (c) Elements in the top row of  Fig.~\ref{fig:TColors} correspond to rotating the tetrahedron by 120 degrees about the corresponding colored vertex, according to the right hand rule.
    (d) Elements of the middle row of Fig.~\ref{fig:TCrossing} correspond to rotating the tetrahedron by 240 degrees about the corresponding colored vertex, according to the left hand rule.
    }
    \label{fig:tetrapic}
\end{figure}

\clearpage

\section{Order parameter space of spin-2 cyclic BEC}

The order parameter space for the cyclic phase of spin-2 BEC is 
\begin{equation}
    M_\mathrm{C} = [S^1 \times \mathrm{SO}(3)] / T,
\end{equation}
where the circle $S^1$ is regarded as the unit complex numbers, and $T$ is the group of the rotational symmetries of a tetrahedron~\cite{semenoff:2007, turner:2007}. It may be realized as the subgroup of $S^1 \times \mathrm{SO}(3)$ generated by the elements $(1, \mathrm{diag}(-1,-1,1)), (1, \mathrm{diag}(-1,1,-1)),$ and $(1, \mathrm{diag}(1,-1,-1))$ and
\begin{equation}
    \left\{\exp\left(\dfrac{2 \pi i}{3} \right), 
    \begin{bmatrix}
    0 & 1 & 0 \\
    0 & 0 & 1 \\
    1 & 0 & 0 
    \end{bmatrix}
    \right\} \in S^1 \times \mathrm{SO}(3).
\end{equation}
Moreover, the fundamental group of $M_\mathrm{C}$ may be realized as a subgroup
\begin{equation}
    \pi_1(M_\mathrm{C}) \subset \mathbb{Z} \times T^*
\end{equation}
where the first component keeps records the fractional scalar-phase winding in multiples of $1/3$. In particular, the composition with the projection map $\mathbb{Z} \times T^* \to T^*$ yields a surjective group homomorphism 
\begin{equation}
    \pi_1(M_\mathrm{C}) \to T^*,
\end{equation}
because for any $\alpha \in T^*$, there exists $n \in \mathbb{Z}$ such that $(n,\alpha) \in \pi_1(M_\mathrm{C})$. Moreover, any two choices of such $n$ differ by an integer multiple of three.

\section{$Q$-invariants}

Here, we define the $Q$-invariants for $T^*$-colored link diagrams, and check their basic properties. In the \emph{conjugate-inverse} equivalence relation of a group $G$, elements $g,h \in G$ are regarded as equivalent if they are conjugate or if $g$ and $h^{-1}$ are conjugate. This relation partitions $G$ into \emph{conjugation-inversion classes}, the elements of which are unions of the conjugacy classes of $g$ and $g^{-1}$ for a $g \in G$. The conjugation-inversion classes correspond to the types of topological vortices.
We will construct three $Q$-invariants, one for each nontrivial conjugation-inversion class of $T^*$ apart from $\{-1\}$ (purple strands). Each invariant is computed by considering all loops of the same type in a $T^*$-colored link diagram. Hereafter, we will ignore the purple strands, because $-1$ commutes with every element of $T^*$ (i.e., $-1$ is a \emph{central element}), and therefore purple loops are not able to link or knot in a topologically protected fashion.

Consider a $T^*$-colored link $\mathcal L$ represented by a $T^*$-colored link diagram consisting of $n$ loops $L_1,\dots,L_n$. Each component, $L_i$, can be equipped with a basepoint, $b_i$, on one of the arcs of the loop in the link diagram. Each loop $L_i$ is oriented in such a fashion that the black or grey color of the bicoloring is on the right side when traversing along a loop starting from the basepoint.

\begin{defn}
Let ($L_i,b_i$) be a loop with a basepoint in a $T^*$-colored link diagram. Let $\chi(b_i)\in T^*$ be the element corresponding to the basepoint. The element $q(L_i,b_i) \in T^*$ is defined as the product of elements corresponding to the non-purple loops $L_i$ crosses under. The order of the product is taken from right to left when traversing the loop, starting from the basepoint, and traversing in the direction of the orientation defined above. An example computation is depicted in Fig.~\ref{fig:Q'def}.
\end{defn}

\begin{lem}\label{lem:qIsPower}
Suppose  
\begin{equation}
    \chi(b_i) \in \left\{\pm i, \pm j, \pm k, \dfrac{1 \pm i \pm j \pm k}{2} \right\}.
\end{equation}
Then the element $q(L_i,b_i)$ is an integer power of $\chi(b_i)$.
\end{lem}
\begin{proof}
For the coloring of the loop $L_i$ to be consistent, $q(L_i,b_i)$ must commute with $\chi(b_i)$. Under the above assumption on $\chi(b_i)$ this implies that $q(L_i,b_i)$ is a power of $\chi(b_i)$. 
\end{proof}

\begin{rem}
The conclusion of Lemma~\ref{lem:qIsPower} is not true for other elements of $T^*$. For example, $q = (1 - i - j - k) / 2$ commutes with $\chi = (-1 + i + j + k)/2$ even though $q$ is not a power of $\chi$. However, if the roles are reversed, then Lemma~\ref{lem:qIsPower} is true: $\chi = q^4$.
\end{rem}

\begin{defn}
Let $(L_i,b_i)$ be as above, and denote by $n_i$ the order of the element $\chi(b_i) \in T^*$. If $\chi(b_i) \in \{\pm i, \pm j, \pm k, (1 \pm i \pm j \pm k) / 2 \}$, then we can define
\begin{equation}\label{eq:Q'}
    Q'(L_i,b_i) := \log_{\chi(b_i)}(q(L_i,b_i)) \in \mathbb{Z}_{n_i},
\end{equation}
where $\log_{\chi(b_i)}$ denotes the \emph{discrete logarithm} with base $\chi(b_i)$. The right side of Eq.~\ref{eq:Q'} is well defined by Lemma~\ref{lem:qIsPower}.
\end{defn}

\begin{lem}\label{lem:basePointIndependence}
$Q'(L_i, b_i)$ is independent of the basepoint $b_i$.
\end{lem}
\begin{proof}
The proof is presented in Fig.~\ref{fig:basepoint}.
\end{proof}

From now on, we will use simpler notation $Q'(L_i) := Q'(L_i,b_i)$.  

\begin{defn}
Let $b\equiv \mathrm{Cl}(\frac{1-i-j-k}{2})$, $g\equiv \mathrm{Cl}(-\frac{1-i-j-k}{2})$ and $c\equiv \mathrm{Cl}(i)$, where $\mathrm{Cl}(x)$ denotes the conjugation-inversion class of the element $x$. More informally, we will refer to the strands of classes $b$, $g$, and $c$ as \emph{black}, \emph{gray}, and \emph{$Q_8$ strands}, respectively. 

We define a bijection $\rho: T^* \to T^*$ by the formula
\begin{equation}
    \rho(a) = 
    \begin{cases}
    a & \text{if $a \in Q_8 \subset T^*$;} \\
    -a  & \text{otherwise.}
    \end{cases}
\end{equation}
The bijection $\rho$ is not a group homomorphism: for example, the degree-three element $(-1 + i + j + k) / 2$ maps to the degree-six element $(1 - i - j - k) / 2$. However, as $-1$ is a central and $T^* \backslash Q_8$ is stable under conjugation, $\phi$ is compatible with conjugations in the sense that
\begin{equation}\label{eq:rhoformula}
    \rho(a b a^{-1}) = \rho(a) \rho(b) \rho(a)^{-1}.
\end{equation}
Hence, $\rho$ induces a symmetry on the set of all $T^*$-colored link diagrams that exchanges the black and the gray strands.
\end{defn}

\begin{defn}[The $Q$-invariants]
For the three conjugation-inversion classes defined above, we define the corresponding $Q$-invariants as
\begin{align}
    Q_b(\mathcal L)&:=\sum_{L_i\text{ of class $b$}} Q'(L_i)-4l_{bg}-\omega_b \in \mathbb{Z}_6 \\
    Q_g(\mathcal L)&:=Q_b(\rho(\mathcal L)) \in \mathbb{Z}_6  \\
    Q_c(b_i)_{L_i\text{ of class $c$}}&:=\sum_{L_i\text{ of class $c$}} Q'(L_i)-\omega_c \in \mathbb{Z}_4 \label{eq:defQc}
\end{align}
Above, the \emph{writhe} term $\omega_x$ is the signed count of all the crossings involving loops of the conjugacy class $x$, where the sign of a right-handed and left-handed crossings are $+$ and $-$, respectively (Fig.~\ref{fig:writhe}). Similarly, $l_{bg}$ is the signed count of all the crossings where black strand crosses under a gray one. 
\end{defn}

For $Q_8$-loops the orientation of the loop depends on the chosen basepoint, which affects the signs of the crossings. Correspondingly, Eq.~\ref{eq:defQc} seems to depend on the basepoints $b_i$ of the $Q_8$-loops $L_i$. However, by the following observation, the invariant $Q_c$ is independent of these choices.

\begin{lem}\label{lem:writhedef}
Let $\cal L$ be a link diagram with loops $L_1,\dots,L_n$. The signed count of crossings of $\cal L$, modulo 4, does not depend on the orientation of the loops $L_i$.
\end{lem}
\begin{proof}
It suffices to prove the following: given an $n$-loop oriented link diagram, the reversal of the orientation of one loop conserves the signed count $\omega$ of the crossings modulo 4. A loop $L_i$ in a link diagram is part of an even number of crossings with other loops. Reversing the orientation of $L_i$ therefore flips the signs of an even number of crossings. As the sign-flip at each crossing changes $\omega$ by $\pm 2$, it follows that the total effect of an even number of crossings leaves $\omega$ invariant modulo 4.
\end{proof}

Next, we want to prove the topological invariance of $Q_b, Q_g$ and $Q_c$. First we show that these invariants depend only on the $T^*$-colored link $\mathcal L$, and not of the specific diagram that is chosen to represent it, by verifying that the invariants are conserved in Reidemeister moves. Then, we show that the invariants are conserved in topologically allowed local surgeries, namely, topologically allowed strand crossings and local reconnections. Thus, the $Q$-invariants may be employed in detecting topologically protected $T^*$-colored links.

\begin{lem}
The $Q$-invariants are conserved under Reidemeister moves.
\end{lem}
\begin{proof}
The proof is presented in Fig.~\ref{fig:Reidemeister}.
\end{proof}

\begin{lem}\label{lem:crossings}
The $Q$-invariants are conserved in topologically allowed local strand crossings.
\end{lem}
\begin{proof}
Let $a \in \{\pm i, \pm j, \pm k \}$, and suppose that $b \in T^*$ either commutes or anticommutes with $a$. Then $b \in Q_8 \subset T^*$. Hence, the proof of Ref.~\cite[Lemma~6]{annala:2022} generalizes to this case. The conservation of $Q_b$ is treated in Fig.~\ref{fig:crossing}, from which the conservation of $Q_g$ follows immediately.
\end{proof}

\begin{lem}\label{lem:reconnections}
The $Q$-invariants are conserved in topologically allowed local reconnections.
\end{lem}
\begin{proof}
As in the proof of Lemma~\ref{lem:crossings}, the proof from Ref.~\cite{annala:2022} generalizes to prove the conservation of $Q_c$. The conservation of $Q_b$ is treated in Fig.~\ref{fig:reconnection}, from which the conservation of $Q_g$ follows immediately.
\end{proof}

Finally, we compare the $Q$-invariants defined here to that defined earlier in Ref.~\cite{annala:2022}.

\begin{prop}
Let $\mathcal L$ be a $T^*$-colored link that consists of $Q_8$-loops. Then 
\begin{equation*}
    Q_c(\mathcal{L}) = Q (\mathcal{L}) \in \mathbb{Z}_4,
\end{equation*}
where $Q(\mathcal{L})$ is the $Q$-invariant of Ref.~\cite{annala:2022}.
\end{prop}
\begin{proof}
In Ref.~\cite{annala:2022}, the $Q$ invariant was defined via the auxiliary \emph{colored invariants} $Q_\mathrm{col}$, where $\mathrm{col}$ is either red, gray, or blue, and, slightly rephrasing the original definition,
\begin{equation}
    Q_{\mathrm{col}} (\mathcal{L}) := \sum_{L_i \text{  of color col}} Q'(L_i)-\omega_\mathrm{col} \text{ (mod 4)},
\end{equation}
where $\omega_\mathrm{col}$ is the signed count of crossings where both strands are of the specified color, which is well defined by Lemma~\ref{lem:writhedef}. For each of the three colors, $Q_{\mathrm{col}}$ obtains the same value, and the invariant $Q(\mathcal{L})$ is defined as the common value. In particular
\begin{align}
    -Q(\mathcal{L}) &= Q_\mathrm{red}(\mathcal{L}) + Q_\mathrm{gray}(\mathcal{L}) + Q_\mathrm{blue}(\mathcal{L}) \\
    &= \sum_{L_i \text{ is a loop of $\cal L$}} Q'(L_i)-\omega',
\end{align}
where $\omega'$ is the signed count of unicolor  crossings of $\cal L$. 

Consider the difference $\omega - \omega' \in \mathbb{Z}_4$, where $\omega$ is the signed count of all crossings of the diagram. It is the signed count of two-color crossings and therefore, by the definition of the linking number, 
\begin{align}
\omega - \omega' &= 2 (l_{rg} + l_{rb} + l_{gb}) 
\end{align}
where $l_{rg}$, $l_{rb}$, and $l_{gb}$ are the total linking numbers between red and gray, red and blue, and gray and blue loops respectively. By Ref.~\cite{annala:2022} 
$l(\mathcal{L}) \equiv l_{rg} \equiv l_{rb} \equiv l_{gb} (\mod 2),$
where the common value $l(\mathcal{L})$ is the \emph{linking invariant} of the diagram. In other words
\begin{equation}
    Q_c(\mathcal{L}) = - Q(\mathcal{L}) + 2l(\mathcal{L}) \in \mathbb{Z}_4.
\end{equation}
As $l(\mathcal{L}) \equiv Q(\mathcal{L}) \mod 2$~\cite{annala:2022}, it follows that 
\begin{align}
    - Q(\mathcal{L}) + 2l(\mathcal{L}) &= - Q(\mathcal{L}) + 2Q(\mathcal{L}) \\
    &= Q(\mathcal{L}) \in \mathbb{Z}_4,
\end{align}
proving the claim.
\end{proof}

\clearpage

\begin{figure}
    \centering
    \includegraphics[scale=1.3]{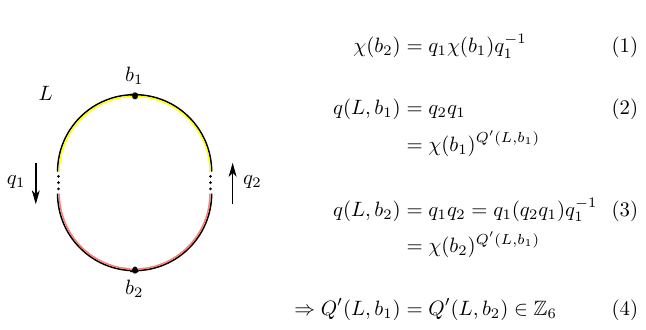}
    \caption{\textbf{Basepoint independence of $Q'$ in $G$-colored link diagrams.} To prove that $Q'(L_i)$ is independent of the basepoint, we consider two different basepoints, $b_1$ and $b_2$, on the loop $L$, and show that $Q'(L, b_1) = Q'(L, b_2)$. Equation~(3) in the figure follows from Eqs.~(1) and (2). It implies that $\chi(b_2)^{Q'(L,b_1)} = \chi(b_2)^{Q'(L,b_2)}$, from which the desired result follows. Even though we treat only the case of black loop in $T^*$-colored link diagram, the proof generalizes, mutatis mutandis, to any $G$-coloring under the assumption that $q_1 q_2$ is a power of $\chi(b_1)$.
    }\label{fig:basepoint}
\end{figure}

\begin{figure}
    \centering
    \includegraphics[scale=1]{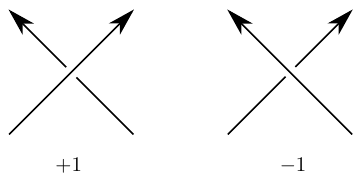}
    \caption{\textbf{Writhe.} Sign of an oriented crossing is determined by the right hand rule.}
    \label{fig:writhe}
\end{figure}

\begin{figure}
    \centering
    \includegraphics[scale=1]{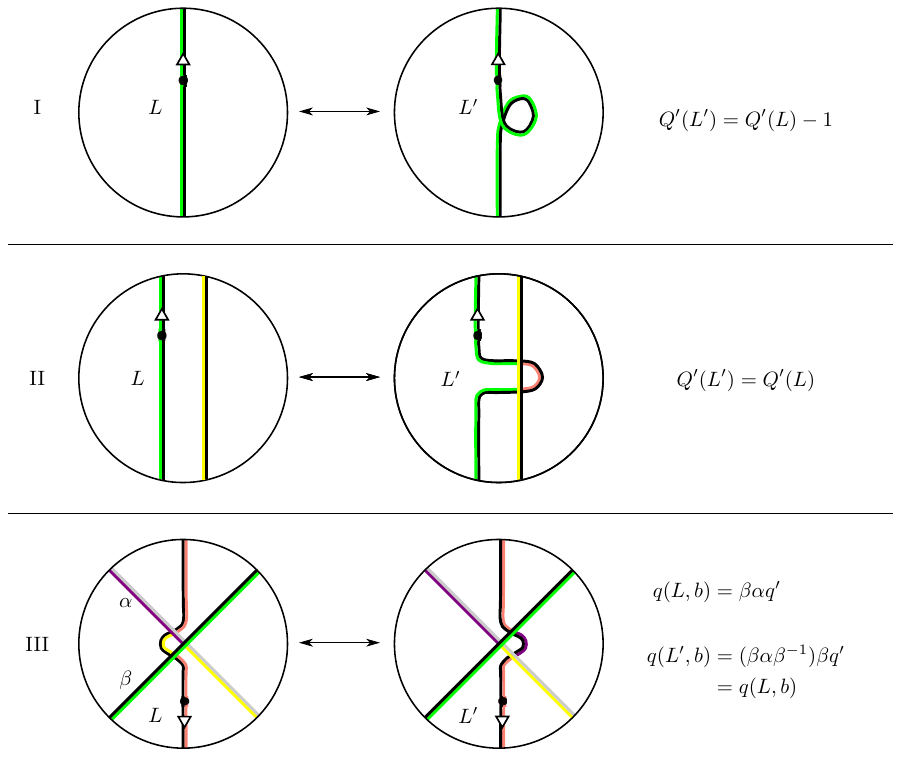}
    \caption{\textbf{Conservation of the $Q$-invariants under the Reidemeister moves.} 
    (\MakeUppercase{\romannumeral 1}) Twisting the loop $L$ in the diagram affects both $Q'(L)$ and the writhe term $\omega$. As these effects cancel each other, the the corresponding $Q$-invariant is conserved. 
    (\MakeUppercase{\romannumeral 2}) Moving the loop $L$ under another affects neither $Q'(L)$ nor the corresponding writhe term $\omega$ (nor $l_{bg}$), thus conserving the $Q$-invariant. (\MakeUppercase{\romannumeral 3}) Moving the loop $L$ under a crossing affects neither $Q'(L_i)$ nor the corresponding writhe term $\omega$ (nor $l_{bg}$), leaving the $Q$-invariant unchanged.
    }\label{fig:Reidemeister}
\end{figure}

\begin{figure}
    \centering
    \includegraphics[scale=1]{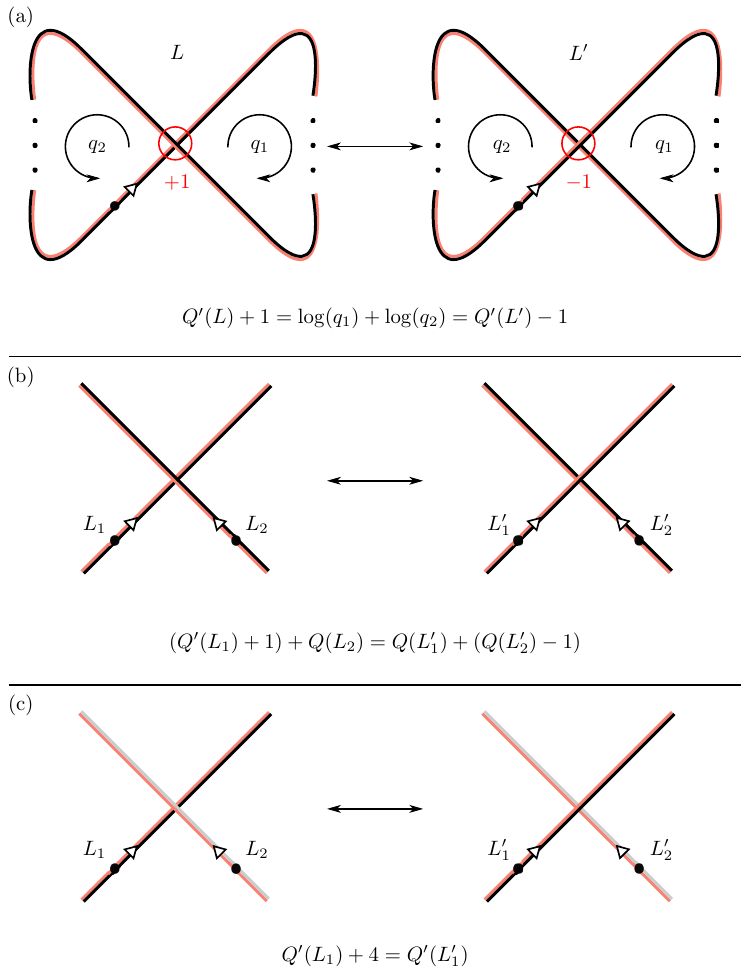}
    \caption{\textbf{Conservation of the invariant $Q_b$ in topologically allowed strand crossings.}
    A black strand may only cross a black strand of the same color or a gray strand of the same color.
    (a) There is essentially one case to consider when both strands of the crossing are part of the same loop $L$. The effects of the strand crossing on $Q'(L)$ and on the writhe term $\omega_b$ cancel each other. The logarithms refer to discrete logarithm with base $(1 - i + j + k) / 2$.
    (b),(c) There are essentially two cases to consider if the strands of the crossing belong to different loops $L_1$ and $L_2$, depending on the class of $L_2$. The effects of the strand crossing on $Q'(L_1)$ (and $Q'(L_2)$ if $L_2$ is a black loop) and on the term $4l_{bg} + \omega_b$ cancel each other.
    }\label{fig:crossing}
\end{figure}

\begin{figure}
    \centering
    \includegraphics[scale=1]{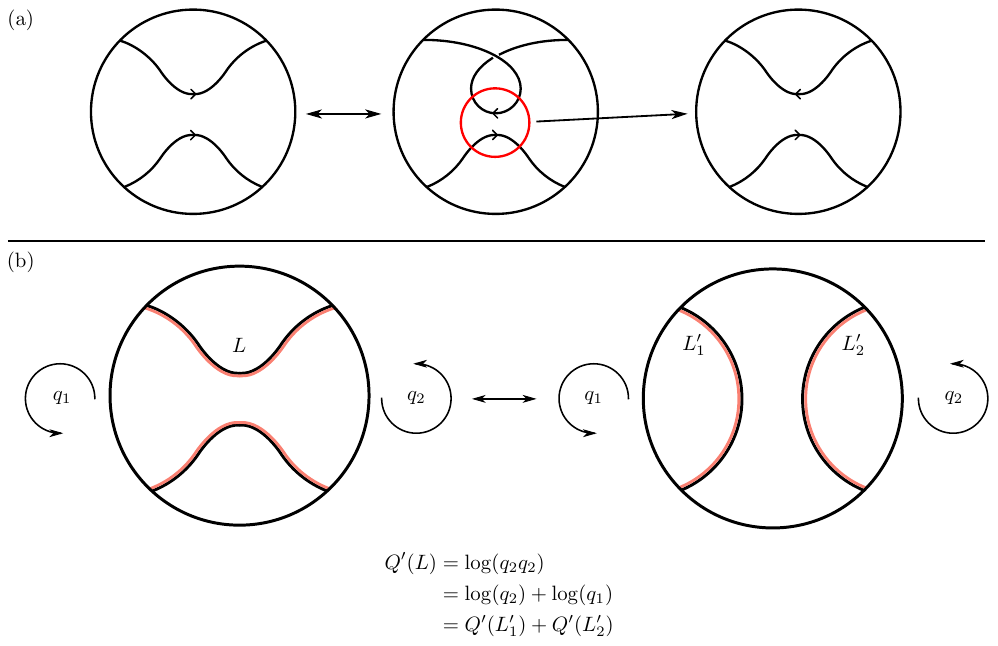}
    \caption{\textbf{Conservation of the invariant $Q_b$ in local reconnections.} 
    Since the orientation of a black loop is conserved in undercrossings, the reconnection always either splits one loop into two loops, which might be linked, although this is not apparent from the diagrammatic presentation, or, alternatively, amalgamates two distinct loops into one loop.
    (a) Up to the first Reidemeister move, we only have to consider the case when the connecting strands have opposite orientations. 
    (b) In the reconnection, the loop $L$ splits into loops $L'_1$ and $L'_2$ in a fashion that conserves the invariant $Q_b$. The logarithms refer to discrete logarithms with base $(1 - i + j + k)/2$
    }
    \label{fig:reconnection}
\end{figure}

\clearpage

\section{Classification}

Here, we classify all $T^*$ colored links that are tied from strands with a coloring in a single conjugacy class. We will focus on the conjugacy class corresponding to the \textit{black strands} (the top row in Fig.~\ref{fig:TColors}), since the gray strands (the middle row in Fig.~\ref{fig:TColors}) behave in an equivalent fashion, and $Q_8$-colored links were already classified in Ref.~\cite{annala:2022}. To simplify the process of classification, we employ several symmetries of the problem, such as planar symmetries of the diagram and color permutations induced by the automorphisms of the group $T^*$, which are depicted in Fig.~\ref{fig:perm1}.

\begin{thm}
Let $L$ be a $T^*$-colored link consisting of black strands. Then $L$ is either topologically trivial or an untangled disjoint union of $T^*$-colored trefoils and figure-eight knots.
\end{thm}
\begin{proof}
It suffices to show that any $T^*$-colored link diagram may be decomposed into an untangled disjoint union of $T^*$-colored trefoil knots, figure-eight knots, and simple loops. By employing the sequence of local surgeries illustrated in Fig.~\ref{fig:cte}(a), each crossing of the diagram of $L$ may be transformed into a double crossing, which we may denote by a red edge as in Fig.~\ref{fig:cte}(b) to simplify the diagram. In this process, a trefoil knot is split off.
In other words, the link diagram of $L$ can be transformed into a disjoint union of simple loops, possibly nested, that are connected only by the edges as shown in Fig.~\ref{fig:cte}(c), up to splitting of trefoil knots.

For each \textit{innermost loop} of the diagram, i.e., a loop that does not contain any other loops, we may eliminate repeating bicolors by splitting the loop into multiple smaller ones. Because there are only four different types of black strands, we may transform the diagram into one, where each innermost loop is connected to at most four edges. Any two-edge loop can be eliminated, as depicted in Fig.~\ref{fig:2loop}, in a fashion that reduces the number of edges. In the process, some loops might amalgamate into large loops, connected to more than four edges, but such loops may again be split into smaller ones, possibly creating again two-edge loops. However, we may repeat the elimination-and-splitting process until there are only three- and four-edged innermost loops because two-edge loop elimination always reduces the number of edges. Hence we are left with only three- and four-edged innermost loops. 


The rest of the proof is separated into two steps.

\

\underline{Step 1.} Elimination of nesting.

\

If the loops of the diagram are nested in non-trivial fashion, then there must exist a \emph{level-two loop} $\alpha$, i.e., a loop inside which every loop is an innermost loop. We now want to transform each loop inside $\alpha$ to have an \textit{orientation} opposite that of $\alpha$, i.e., if the black color of the bicoloring of $\alpha$ is on the outside, then all the loops inside $\alpha$ should have it on the inside, and vice versa (Fig.~\ref{fig:regions}). This may be achieved by first splitting every four-edged loop with the wrong orientation into three- and two-edged loops as in Fig.~\ref{fig:43}(a), and then flipping the orientation of each three-edged loop with the wrong orientation as demonstrated in Fig.~\ref{fig:43}(b). After this, the two-edged loops may be eliminated using the process of Fig.~\ref{fig:2loop}, which might also affect the loop $\alpha$. Nonetheless, after the two-edge-loop elimination, the remnants of $\alpha$ only contain innermost loops with an orientation opposite that of $\alpha$. Each edge inside $\alpha$ is four-colored as illustrated in Fig.~\ref{fig:EdgeTypes}. By employing surgery operations depicted in Fig.~\ref{fig:edge}(b), we may eliminate the nesting, as depicted in Fig.~\ref{fig:regions}.

In such a fashion, we may eliminate all level-two loops, producing a diagram consisting of disjoint union of non-nested simple loops, connected by edges. We may reduce to the case where all loops are either three- or four-edged, and all of them have the same orientation. This implies that all edges of the diagram are four colored. This is the end of the first step.

\

\underline{Step 2.} Decomposition of a diagram without nesting.

\

The diagram divides the plane into several ``polyhedral regions'', the edges of which are the edges of the diagram, and the vertices of which are the loops (\emph{vertex loops}). If there exists a repeating color inside such a region, it may be broken apart, without altering the number of edges in the diagram. Hence, we may assume that there exists at least one region that is bounded by two, three, or four edges. Moreover, by pruning the vortex loops, we may assume that either all of them are connected to at most four edges, or that we can immediately reduce number of edges in the diagram by at least one (Fig.~\ref{fig:pruning}). It follows from the classification of four-colored edges (Fig.~\ref{fig:edgeclass}) that, up to even color permutations and mirroring of the plane, all bicolorings of such regions are equivalent. We will show that in the neighbourhood of such diagram, one can perform local surgery operations that reduce the number of edges in the diagram. 

There are three cases to consider:
\begin{enumerate}
    \item A region bounded by two edges can be eliminated (Fig.~\ref{fig:2elimination}), thus reducing the number of loops edges.

    \item The number of edges around a region bounded by three edges can be reduced (Fig.~\ref{fig:3elimination1} and Fig.\ref{fig:3elimination2}).

    \item The number of edges around a region bounded by four edges can be reduced (Fig.~\ref{fig:4elimination1}, Fig.~\ref{fig:4elimination2}, and Fig.\ref{fig:4elimination3})
\end{enumerate}

In each case, the number of edges in the diagram can be reduced by a process that possibly splits off $T^*$-colored trefoil and figure-eight knots. If the number of edges is 0, then the diagram describes an untangled disjoint union of loops. This proves that any $T^*$-colored link may be decomposed into an untangled disjoint union of $T^*$-colored trefoil knots, figure-eight knots, and simple loops, as desired. \qedhere



\end{proof}

\clearpage

\begin{figure}
    \centering
    \includegraphics[scale=1.3]{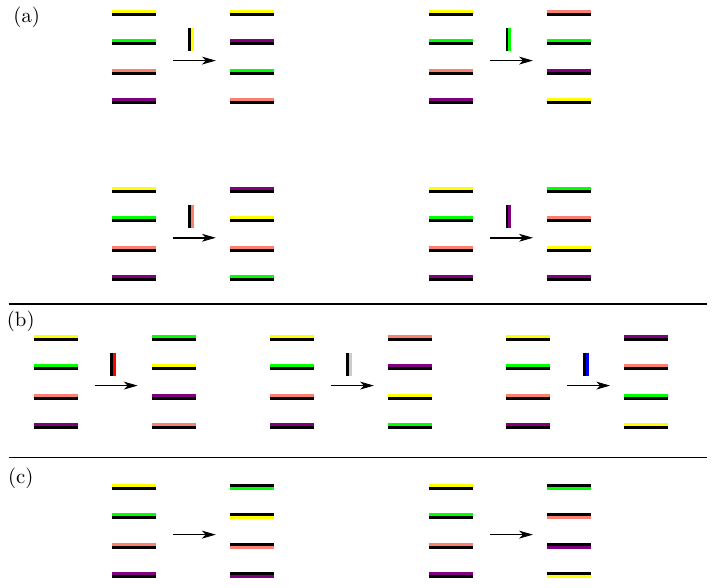}
    \caption{
    \textbf{Automorphisms of $T^*$ and the induced color permutations on black strands.} An automorphism of the binary tetrahedral group $T^*$ induces a permutation of bicolorings on a $T^*$-colored link diagram. This is useful in relating different bicolorings of a diagram to each other. The inner automorphism group of $T^*$ is isomorphic to the group $A_4$ of even permutations on four letters (the \emph{alternating group}). They are obtained by conjugating with either black strands (a) or $Q_8$ strands (b). Concretely, the permutation of colors corresponds to the permutation of vertices of the tetrahedron induced by the corresponding rotational tetrahedral symmetry. The full automorphism group of $T^*$ is isomorphic to the symmetric group $S_4$~\cite{golasiski:2011}. An odd permutation of colors flips the bicoloring, as illustrated by two examples (c). Concretely, the permutation of colors corresponds to the permutation of the vertices of the tetrahedron induced by the corresponding tetrahedral symmetry, which is not necessarily rotational. The flipping of the bicoloring is explained by the fact that an orientation reversing symmetry flips the handedness of the $120^\circ$ rotation that corresponds to the bicoloring.
    }
    \label{fig:perm1}
\end{figure}

\begin{figure}
    \centering
    \includegraphics[scale=0.9]{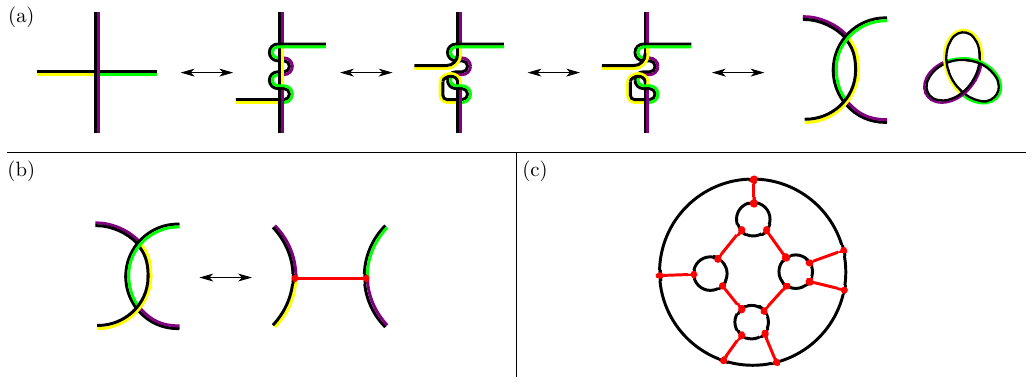}
    \caption{
    (a) By employing a sequence of local surgeries, a crossing can be transformed into a double crossing by splitting of a $T^*$-colored trefoil knot. Up to color permutations (Fig.~\ref{fig:perm1}) and planar symmetries, there is only one case to consider. Hence the result stands for all $T^*$-colored crossings. 
    (b) A double crossing can be depicted by an edge in order to simplify the presentation. 
    (c) By performing the local surgery in (a) for all crossings on the diagram of $L$, the diagram may be transformed into a set of untangled, possibly nested simple loops connected by edges. 
    }
    \label{fig:cte}
\end{figure}


\begin{figure}
    \centering
    \includegraphics[scale=1.1]{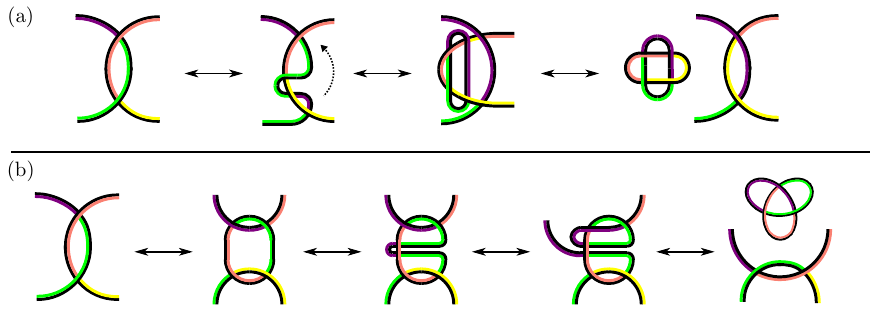}
    \caption{
    \textbf{Basic surgeries, I.}
    (a) The order of crossings may be swapped in a four-colored edge by a surgery that splits off a $T^*$-colored Solomon's link. Moreover, Solomon's link may be split into a disjoint union of a $T^*$-colored trefoil knot and a $T^*$-colored figure-eight knot as shown in Fig.~\ref{fig:examples}(c).
    (b) The orientation of a four-colored edge may be altered by a surgery that splits of a $T^*$-colored trefoil.
    }\label{fig:edge}
\end{figure}

\begin{figure}
    \centering
    \includegraphics[scale=1.25]{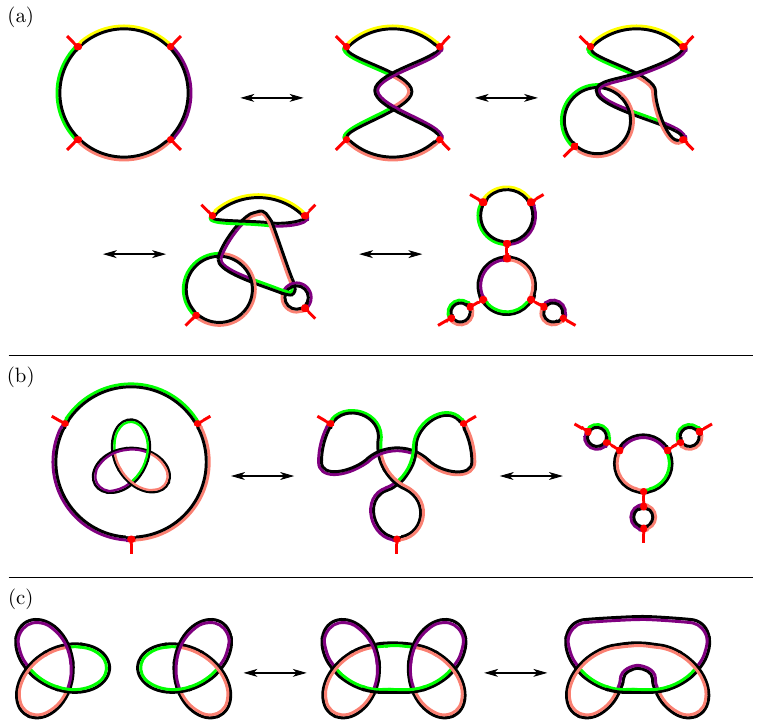}
    \caption{
    \textbf{Basic surgeries, II.}
    (a) Surgery splitting a four-edge loop into two three-edge loops and two two-edge loops. 
    (b)~Orientation of a three-edge loop may be flipped by employing a trefoil. One can always produce a pair of left- and right handed trefoils from topologically trivial configuration (c). Consequently, instead of regarding the orientation flip as annihilating a trefoil knot, one may regard it as splitting off the inverse trefoil.
    }
    \label{fig:43}
\end{figure}

\begin{figure}
    \centering
    \includegraphics[scale=1.1]{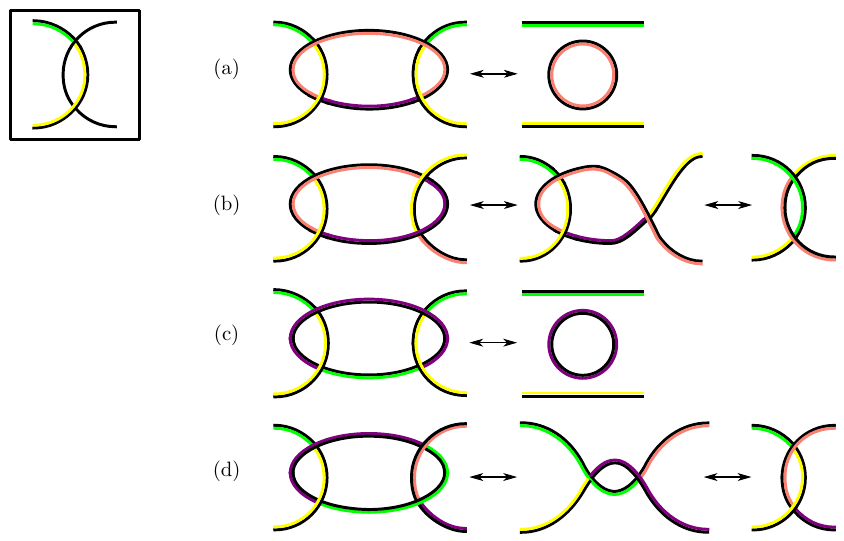}
    \caption{
    \textbf{Elimination of two-edge loops.} Up to even color permutations and planar symmetry, we may fix the colors of the leftmost loop, and the order of over- and undercrossings at the leftmost double crossing. This leaves four cases up to planar symmetry and the local surgery depicted in Fig.~\ref{fig:edge}(a): 
    (a) both edges are four colored (one case);
    (b) one edge is four-colored and the other is three-colored (one case);
    (c),(d) both edges are three-colored (two cases).
    In each case, the two-edge loop may be eliminated by employing the surgeries of Fig.~\ref{fig:cte}(a) and Fig.~\ref{fig:edge}(b). The $T^*$-colored trefoil and figure-eight knots that are split off in the surgeries are omitted above for simplicity.
    }\label{fig:2loop}
\end{figure}

\begin{figure}
    \centering
    \includegraphics[scale=2]{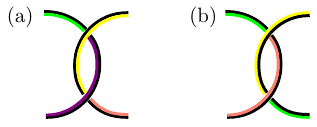}
    \caption{
        \textbf{Edge types.}
        The double crossings are classified according to the orientations of the crossing vortices. (a) If the strands have opposite orientation, then there are four colors adjacent to the double crossing (four-colored edge). (b) If the strands have the same orientation, then there are three colors adjacent to the double crossing (three-colored edge). In both cases, there is only one case to consider up to color permutations depicted in Fig.~\ref{fig:perm1}, and planar symmetries. 
    }
    \label{fig:EdgeTypes}
\end{figure}

\begin{figure}
    \centering
    \includegraphics[scale=1.2]{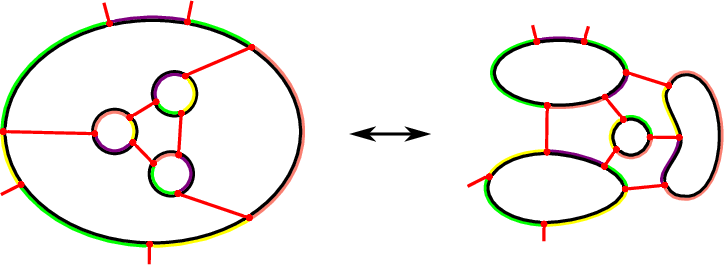}
    \caption{
    \textbf{Elimination of nesting.} 
    A level-two loop (a loop inside which no loop contains another loop) with an opposite orientation to that of the loops it contains can be locally decomposed into a non-nested diagram by repeatedly applying the local surgery of Fig.~\ref{fig:edge}(b). This removes a layer of nesting locally. 
    }\label{fig:regions}
\end{figure}


\begin{figure}
    \centering
    \includegraphics[scale=1.5]{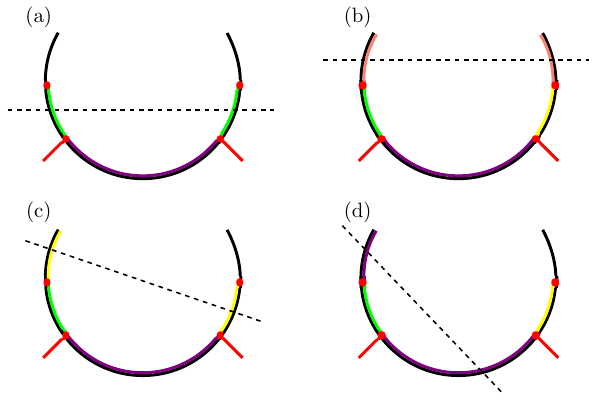}
    \caption{
    \textbf{Pruning vertex loops of a region.} If a vertex loop of a region has more than four edges, then a color is repeated inside it and it can be split into multiple pieces. There are essentially four cases to consider: in (a), (b), and (c) the loop splits in such a fashion that the new vertex loop will have at most four edges. The fourth case (d) produces a two-edge loop, which can be eliminated (Fig.~\ref{fig:2loop}), thus reducing the number of edges in the diagram. 
    }
    \label{fig:pruning}
\end{figure}

\begin{figure}
    \centering
    \includegraphics[scale=1.5]{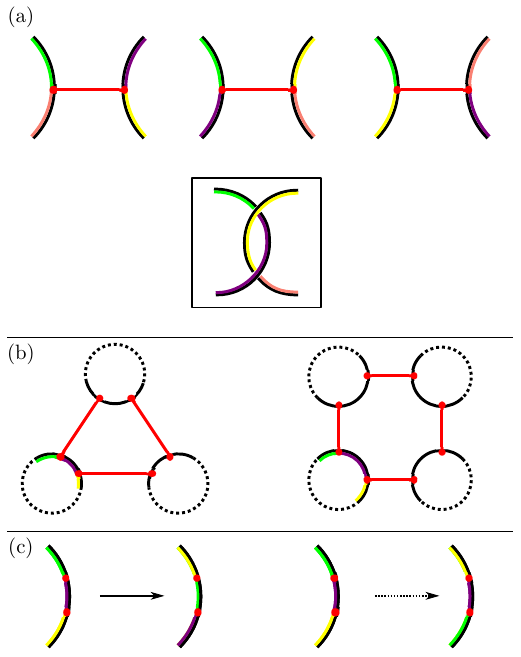}
    \caption{
    \textbf{Classification of four-colored edges and consequences.}
    (a) There are exactly three types of four-colored edges, which are equivalent to one another under the automorphisms of $T^*$ (Fig.~\ref{fig:perm1}). The panel displays the middle edge in detail.
    (b) The coloring on the bottom left loop extends to the whole diagram in such a way that each edge is four colored. The extension is unique for the left diagram, and for the right diagram it is unique under the assumption that no color is repeated inside the region.
    (c) All colorings without repeated colors on the bottom left loop are equivalent under even color permutations (solid arrow, Fig.~\ref{fig:perm1}) and planar mirrorings (dashed arrow).
    }
    \label{fig:edgeclass}
\end{figure}


\begin{figure}
    \centering
    \includegraphics[scale=2]{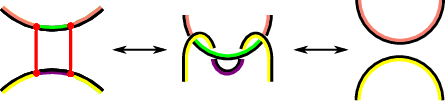}
    \caption{
    \textbf{Elimination of regions bounded by two edges.} Any region bounded by two edges is topologically trivial, as depicted above. Note that the order of the double crossings corresponding to the edges is arbitrary, up to splitting off a $T^*$-colored trefoil and a figure-eight knot (Fig.~\ref{fig:edge}).
    }
    \label{fig:2elimination}
\end{figure}

\begin{figure}
    \centering
    \includegraphics[scale=1.25]{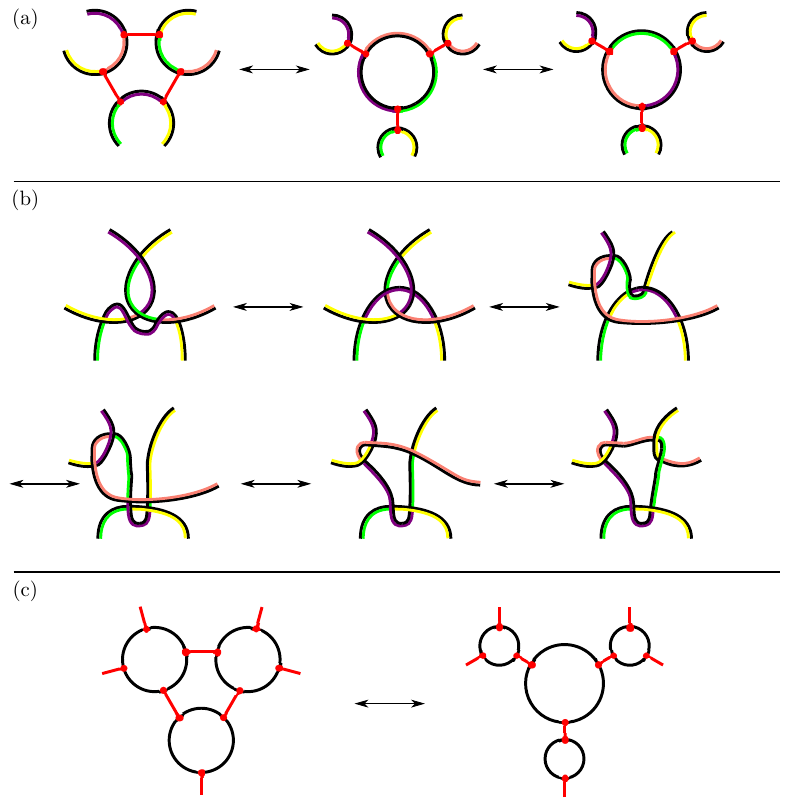}
    \caption{
    \textbf{Elimination of regions bounded by three edges, I.} 
    (a) Diagrammatic sketch of the result of surgeries on a region bounded by three edges. The first transformation is shown in detail in (b), and after eliminating the two-edge loops, the second transformation is obtained from an orientation flip as in Fig.~\ref{fig:43}(b).
    (c) If at least one of the vertex loops has only three edges, we may eliminate the resulting two-edge loop in order to reduce the number of edges in the diagram. 
    }
    \label{fig:3elimination1}
\end{figure}

\begin{figure}
    \centering
    \includegraphics[scale=1.25]{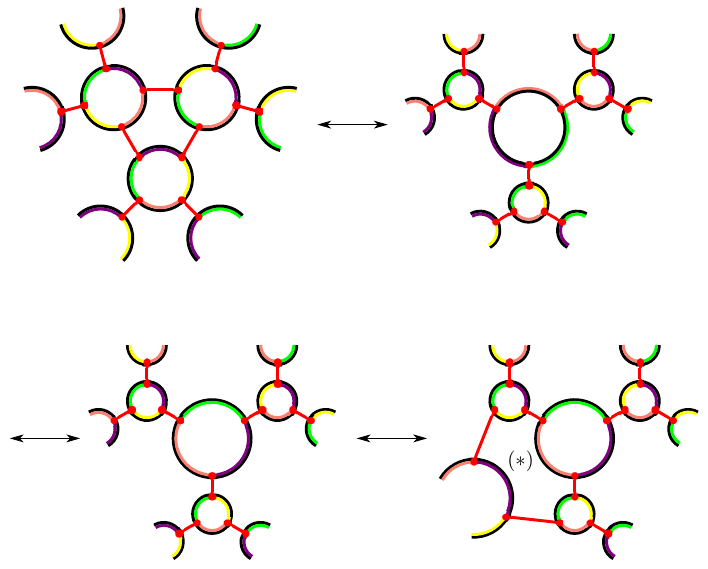}
    \caption{
    \textbf{Elimination of regions bounded by three edges, II.} 
    If all of the vertex loops of a region bounded by three edges have four edges, then we may first perform the surgery depicted in Fig.~\ref{fig:3elimination1}, flip the orientation of the middle loop according to Fig.~\ref{fig:43}(b), and eliminate the resulting two-edge loops, after which we may form a region bounded by four edges $(*)$. The edges around this region can be decreased by the way of Fig.~\ref{fig:4elimination2}(a), leading to a decrease in the number of edges in the diagram.
    }
    \label{fig:3elimination2}
\end{figure}

\begin{figure}
    \centering
    \includegraphics[scale=1]{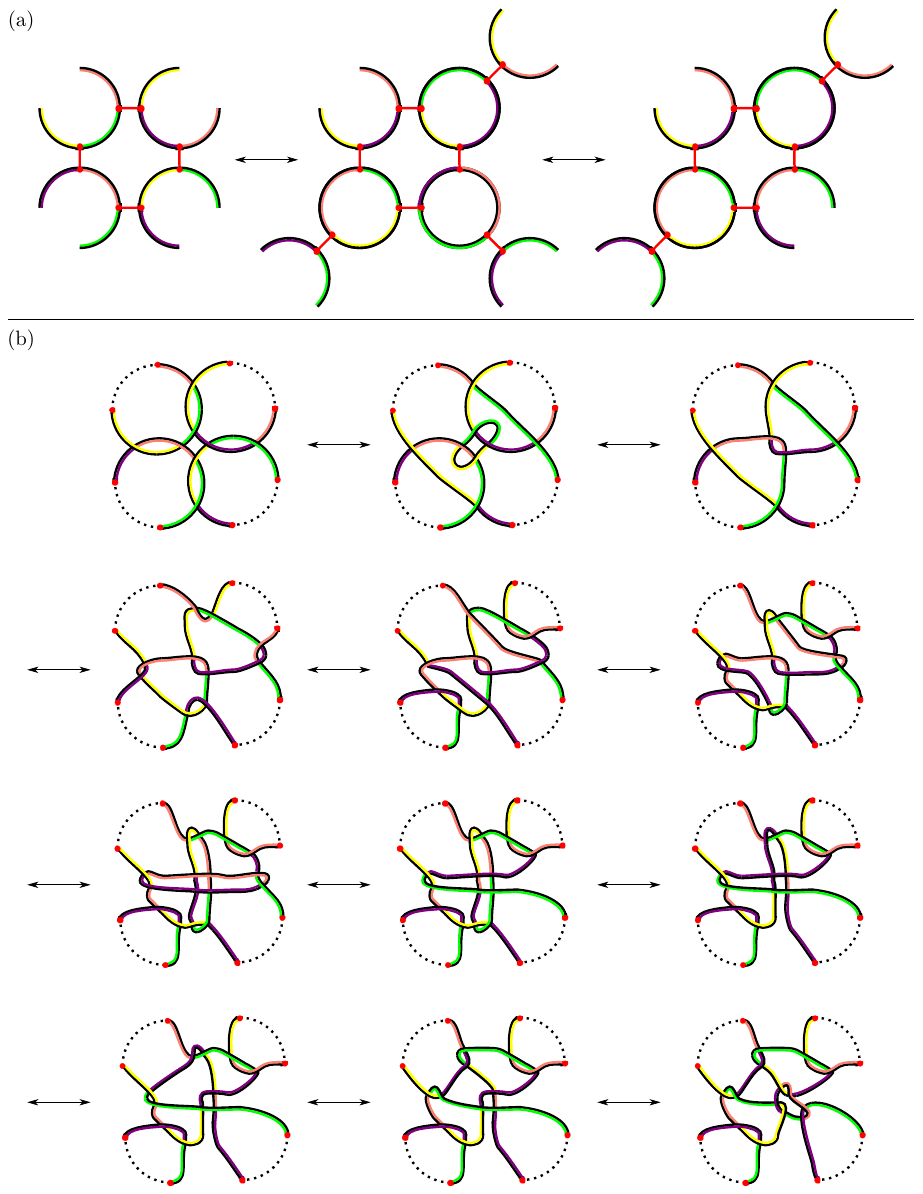}
    \caption{
    \textbf{Elimination of regions bounded by four edges, I.}
    (a) Diagrammatic sketch of the surgery performed inside a region bounded by four edges. The first transformation is depicted in detail in (b). The second transformation is obtained from an orientation flip illustrated in Fig.~\ref{fig:43}(b) after eliminating two-edge loops.
    }
    \label{fig:4elimination1}
\end{figure}

\begin{figure}
    \centering
    \includegraphics[scale=0.8]{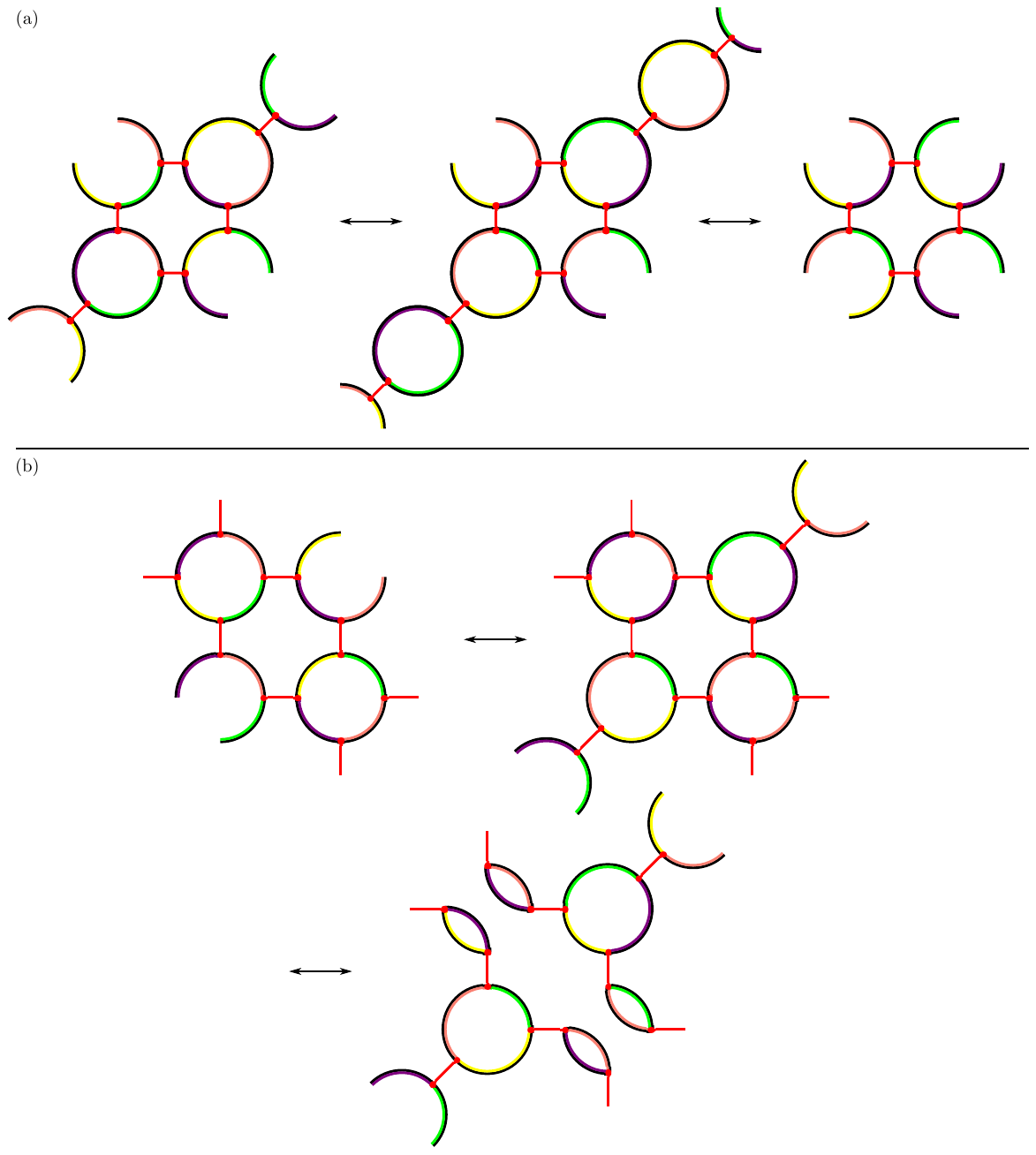}
    \caption{
    \textbf{Elimination of regions bounded by four edges, II.}
    (a) If there are three-edge vertex loops on the opposite sides of the region, the number of edges may be decreased via the procedure presented above.
    (b) If there are four-edge vertex loops on the opposite sides of the region, the number of edges may be decreased via the procedure presented above, and by eliminating the resulting two-edge loops. 
    }
    \label{fig:4elimination2}
\end{figure}

\begin{figure}
    \centering
    \includegraphics[scale=1.1]{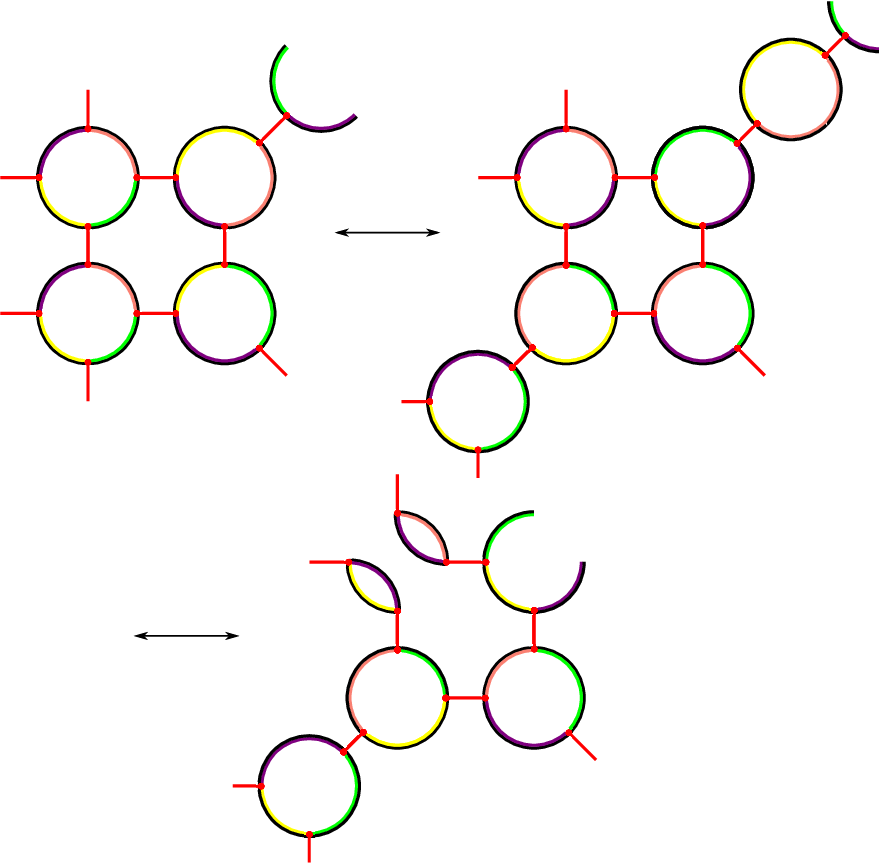}
    \caption{
    \textbf{Elimination of regions bounded by four edges, III.}
    If there does not exist a pair of three- or four-edge vertex loops on the opposite sides of the region, then the number of edges in the diagram can be decreased by the procedure presented above, and by eliminating the resulting two-edge loops.
    }
    \label{fig:4elimination3}
\end{figure}


\end{document}